\documentclass[preprintnumbers,amsmath,amssymb,floatfix,10pt,prd,twocolumn,superscriptaddress,nofootinbib]{revtex4}

\usepackage[parfill]{parskip}
\usepackage{graphicx}
\usepackage{xcolor}
\usepackage{float}
\usepackage{orcidlink}
\usepackage{adjustbox}
\usepackage{bm}
\usepackage{mathtools}
\usepackage{amsthm}
\usepackage{epstopdf}
\usepackage{hyperref}
\newtheorem{theorem}{Theorem}
\newtheorem{lemma}{Lemma}
\newtheorem{proposition}{Proposition}
\newtheorem{corollary}{Corollary}
\theoremstyle{remark}
\newtheorem{remark}{Remark}

\newcommand{\added}[1]{{#1}}

\begin{document}

\title{Horizon formation from effective matter profiles in static spacetimes}

\author{Pablo León \orcidlink{0000-0002-7104-5746}}
\email{pleot@uc.cl}
\affiliation{Instituto de Física, Pontificia Universidad Católica de Chile, Av. Vicuña Mackenna 4860, Santiago, Chile.}

\author{G. Alencar}
\email{geova@fisica.ufc.br}
\affiliation{Departamento de Física, Universidade Federal do Ceará, Fortaleza, Ceará, Brazil}

\author{Manuel Gonzalez-Espinoza
 \orcidlink{0000-0003-0961-8029}}
\email{manuel.gonzalez@pucv.cl}
\affiliation{
Instituto de F\'{\i}sica, Pontificia Universidad Cat\'olica de 
Valpara\'{\i}so, 
Casilla 4950, Valpara\'{\i}so, Chile.}

\author{Francisco Tello-Ortiz \orcidlink{0000-0002-7104-5746}}
\email{francisco.tello@ufrontera.cl}
\affiliation{Departamento de Ciencias Físicas, Universidad de La Frontera, Casilla 54-D, 4811186 Temuco, Chile.}

\begin{abstract}
{The formation of event horizons is traditionally studied as the endpoint of gravitational collapse, where the matter distribution and the spacetime geometry are evolved simultaneously toward a black hole state. In this work, we consider the inverse problem: given a static geometry containing the simplest causally exposed singular structure, namely a timelike naked singularity, what are the minimal conditions on the surrounding matter required for the emergence of an event horizon?
Within classical general relativity, we derive sufficient conditions for horizon formation in terms of the radial organization, compactness and finiteness of the matter distribution. These conditions are summarized by a simple geometric criterion that determines whether a static configuration becomes causally inaccessible to external observers. We further identify situations in which horizon formation necessarily fails, thereby characterizing both the existence and obstruction of causal cloaking in static spacetimes.
Our results show that the emergence of an event horizon is not controlled solely by the total amount of matter, but rather by the way mass-energy is radially accumulated. This provides a minimal and geometrically transparent framework for understanding horizon formation as a causal transition in static spacetimes, independent of gravitational collapse or dynamical evolution.
}
\end{abstract}

\maketitle

%========================================================

\section{Introduction }

The appearance of spacetime singularities remains one of the deepest conceptual problems of general relativity, challenging the predictive character of the theory itself. Under broad and physically reasonable assumptions, the singularity theorems of Hawking and Penrose imply that singular behavior is generically unavoidable in gravitational collapse and cosmology \cite{HawkingEllis1973,Wald1984,Penrose:1964wq}. In many situations these singularities are hidden behind event horizons and therefore remain causally disconnected from asymptotic observers. However, Einstein's equations also admit solutions containing naked singularities, namely singular regions not protected by horizons and therefore causally accessible from infinity \cite{Eardley1978,Joshi2007,JoshiMalafarina2011,Gundlach2025}. Such configurations are particularly problematic because they potentially spoil deterministic evolution and the global predictability of spacetime.

Beyond their physical interpretation, naked singularities offer a useful theoretical setting in which to investigate more general questions concerning causal accessibility and the emergence of event horizons. In particular, they provide the simplest examples of static geometries in which the causal shielding usually associated with black holes is absent. This naturally leads to the following question: what are the minimal conditions required for an event horizon to emerge around a causally exposed region?

The possibility that physically realistic gravitational collapse should never produce naked singularities motivated Penrose's cosmic censorship conjecture \cite{Penrose:1969pc}. Despite decades of work, no general proof of the conjecture exists, and numerous collapse scenarios exhibiting naked singularities have been reported in the literature \cite{Joshi2007,Christodoulou1984,Christodoulou1999}. At the same time, several approaches have been developed to either regularize singular geometries or shield them behind horizons. These include regular black hole constructions \cite{Bardeen1968,Hayward2006,Dymnikova1992}, gravastar and shell-like models \cite{MazurMottola2004,VisserWiltshire2004}, and semiclassical dressing mechanisms \cite{Casals2016}. Most of these approaches focus on constructing particular geometries sourced by specific matter sectors.

The purpose of the present work is different. Rather than constructing a new black-hole solution or proposing a particular regularization mechanism, we investigate the inverse problem of horizon formation. Instead of specifying a matter sector and studying whether gravitational collapse leads to a black hole, we consider the opposite question: given a static horizonless geometry, what are the minimal properties that a matter distribution must satisfy in order to generate an event horizon? More specifically, we seek sufficient conditions guaranteeing the emergence of a single outer horizon and a causally simple exterior region.

This formulation is deliberately independent of the microscopic origin of the source. The same conditions may be tested for classical anisotropic fluids, effective matter sectors induced by additional fields, semiclassical corrections, or phenomenological stress tensors arising from modified-gravity models once they are written in Einstein form. Our emphasis is therefore not on the detailed nature of the source, but on the relation between matter organization and causal structure.

To address this problem, we employ the negative-mass Schwarzschild geometry as a minimal and analytically transparent static horizonless background. Although negative mass is itself exotic, the spacetime is asymptotically flat, possesses a timelike singularity at the origin, and most importantly lacks an event horizon. In the present work, this geometry is used only as a theoretical probe. We do not assume that negative-mass configurations are physically realized, nor do we regard the vacuum negative-mass solution as a realistic endpoint of gravitational collapse. Its role is instead diagnostic: it removes inessential complications and allows one to isolate, in the cleanest possible setting, the matter properties responsible for horizon emergence. It is worth mentioning that this solution has been proven to be unstable \cite{Gleiser:2006yz}. This instability is consistent with our interpretation of the seed metric as a theoretical laboratory rather than as a stable physical object.

The usual direction in gravitational-collapse studies is forward: one specifies matter content and initial conditions and subsequently determines whether collapse leads to a black hole or to a naked singularity \cite{Joshi2007,JoshiMalafarina2011}. Here we adopt the opposite viewpoint. Starting from an already existing horizonless geometry, we ask which features of a static matter distribution are sufficient to generate an event horizon. In this sense, the problem addressed here should not be viewed as an attempt to hide a pathology, but rather as an investigation of the minimal causal mechanism responsible for horizon formation in static spacetimes.

In this work we analyze the problem entirely within classical general relativity. Assuming a static and spherically symmetric anisotropic source, we solve Einstein's equations exactly and express the geometry in terms of a general density profile $\rho(r)$. The horizon structure can then be reduced to the global analysis of an auxiliary function $\Phi(r)=2m(r)-r$,
whose zeros determine the existence and multiplicity of horizons. This allows us to formulate the problem geometrically in terms of the extrema structure of $\Phi(r)$ and the behavior of the quantity $r^{2}\rho(r)$.

Our main result is a theorem providing sufficient conditions for horizon formation. Roughly speaking, we show that if the density profile is non-negative, sufficiently localized, monotonic in an appropriate sense, and large enough to compensate the causal deficit of the seed geometry, then the resulting spacetime develops a unique outer event horizon with a causally simple exterior. Conversely, violations of these conditions naturally lead to multi-horizon geometries or prevent horizon formation altogether. Thus, the theorem identifies which properties of the matter sector are genuinely responsible for the emergence of horizons and which merely affect the detailed geometry.

The paper is organized as follows. In Sect.~\ref{sec2} we derive the Einstein equations for a negative-mass Schwarzschild geometry and obtain the general metric associated with an arbitrary static density profile. In Sect.~\ref{sec3} we analyze the global structure of the horizon equation and reformulate the causal problem in terms of the auxiliary function $\Phi(r)=2m(r)-r$. In Sect.~\ref{sec4} we formulate and prove the Sufficient horizon-formation conditions theorem, establishing sufficient conditions for the existence of a unique outer event horizon together with the corresponding failure modes. In Sect.~\ref{sec5} we study representative dressing profiles illustrating different causal behaviors, including asymptotically Schwarzschild layers, critical logarithmic branches, RN-like geometries, and smooth T-duality-inspired dressings. Finally, in Sect.~\ref{sec6} we discuss the physical implications of the results.

%========================================================
\section{Einstein equations for negative-mass Schwarzschild geometry}
\label{sec2}
%========================================================

We consider a static and spherically symmetric line element in Schwarzschild-like coordinates,
\begin{equation}
ds^2=-f(r)\,dt^2+\frac{dr^2}{f(r)}+r^2 d\Omega^2,
\label{metric}
\end{equation}
with \(d\Omega^2=d\theta^2+\sin^2\theta\,d\phi^2\). We assume an anisotropic matter source
\begin{equation}
T^\mu{}_\nu=\text{diag}\left(-\rho(r),\,p_r(r),\,p_t(r),\,p_t(r)\right).
\label{stress}
\end{equation}
For the metric ansatz \eqref{metric}, Einstein's equations imply the standard mass-function representation
\begin{equation}
f(r)=1-\frac{2m(r)}{r},
\label{frmass}
\end{equation}
where the mass function satisfies
\begin{equation}
m'(r)=4\pi r^2 \rho(r).
\label{mprime}
\end{equation}
Integrating,
\begin{equation}
m(r)=m_0+4\pi\int_0^r \rho(s)s^2 ds,
\label{mint}
\end{equation}
with \(m_0\) an integration constant. To reproduce the negative-mass Schwarzschild behavior near the origin in the absence of dressing matter, we set
\begin{equation}
m_0=-|M|, \qquad |M|>0.
\label{m0}
\end{equation}
\added{This choice should be understood as a minimal seed geometry for the horizon-generation problem. The negative-mass Schwarzschild solution is not introduced here as a realistic physical object, nor as a proposed endpoint of gravitational collapse. Rather, it is the simplest static, asymptotically flat, horizonless geometry with a timelike singularity, and therefore provides a clean background in which the role of the additional matter sector in generating horizons can be isolated.}
Hence,
\begin{equation}
m(r)=-|M|+4\pi\int_0^r \rho(s)s^2 ds.
\label{mfinal}
\end{equation}
The resulting lapse function is
\begin{equation}
f(r)=1+\frac{2|M|}{r}-\frac{8\pi}{r}\int_0^r \rho(s)s^2 ds.
\label{ffinal}
\end{equation}

It is useful to isolate the matter contribution by defining
\begin{equation}
\Delta m(r)\equiv 4\pi\int_0^r \rho(s)s^2 ds,
\qquad
m(r)=-|M|+\Delta m(r).
\end{equation}
Then
\begin{equation}
f(r)=1-\frac{2}{r}\left[-|M|+\Delta m(r)\right].
\end{equation}
This expression already makes clear the mechanism behind horizon generation: the bare negative mass contributes positively to \(f(r)\) through the term \(+2|M|/r\), so the negative-mass Schwarzschild geometry has no horizon, while a sufficiently large positive matter contribution can reduce \(f(r)\) and drive it through zero. \added{In this formulation, the central question is therefore not whether the seed spacetime is itself physically realized, but which minimal properties of the effective matter distribution are required for the metric function to develop zeros.}

{
The radial and tangential pressures follow from the remaining Einstein equations. For the metric ansatz \eqref{metric}, one has \(G^t{}_t=G^r{}_r\), which implies
\begin{equation}
p_r(r)=-\rho(r).
\label{pr}
\end{equation}
The tangential pressure is
\begin{equation}
8\pi p_t(r)=-\frac{m''(r)}{r}
=-\frac{1}{r}\frac{d}{dr}\left(4\pi r^2\rho(r)\right)
=-4\pi\left(2\rho+r\rho'\right).
\label{pt}
\end{equation}
Equations \eqref{mprime}--\eqref{pt} show that once \(\rho(r)\) is specified, the full source is fixed. \added{The source may be interpreted at this stage as an effective stress-energy tensor. Its microscopic origin is not fixed by the formalism: it may correspond to a classical anisotropic distribution, to additional matter fields, or to effective corrections whose net contribution can be written in the form \eqref{stress}.}
}

A first necessary condition for physical admissibility is asymptotic flatness. From \eqref{ffinal} it follows that this requires the matter sector to contribute a finite total mass,
\begin{equation}
\Delta M\equiv \lim_{r\to\infty}\Delta m(r)=4\pi\int_0^\infty \rho(s)s^2 ds <\infty.
\label{deltaM}
\end{equation}
When \eqref{deltaM} holds,
\begin{equation}
m(r)\to M_{\rm eff}\equiv -|M|+\Delta M,
\qquad
f(r)\sim 1-\frac{2M_{\rm eff}}{r},
\label{asymp}
\end{equation}
as $r\to+\infty$.
Thus, the dressed spacetime is asymptotically Schwarzschild with ADM mass \(M_{\rm eff}\). \added{The horizon-generation problem is then controlled by how the accumulated matter contribution \(\Delta m(r)\) competes with the negative seed contribution at finite radius, not only by the value of the total ADM mass.}

The sign of \(M_{\rm eff}\) will play a central role. If \(M_{\rm eff}\le 0\), then the asymptotic region does not resemble a positive-mass black hole exterior. If \(M_{\rm eff}>0\), there is at least the possibility of black-hole-like horizon generation. However, positivity of the asymptotic mass alone is not enough; one must also control the number of zeros of \(f(r)\). \added{Consequently, the relevant classification problem is to determine which density profiles generate a single outer zero of \(f(r)\), which profiles generate multiple horizons, and which profiles fail to generate horizons altogether.}

%========================================================

\section{Global horizon analysis}
\label{sec3}
%========================================================

The horizon equation is
\begin{equation}
f(r_H)=0
\quad \Longleftrightarrow \quad
2m(r_H)=r_H.
\label{horizon}
\end{equation}
It is advantageous to define the auxiliary function
\begin{equation}
\Phi(r)\equiv 2m(r)-r
=-2|M|+8\pi\int_0^r \rho(s)s^2 ds-r.
\label{Phi}
\end{equation}
Then horizons are precisely the zeros of \(\Phi\)
\begin{equation}
\Phi(r_H)=0.
\end{equation}
\added{Thus, the question of horizon generation is reduced to the global behavior of a single real function. This is the main reason for using the negative-mass Schwarzschild seed: it leads to the simplest possible competition between a horizonless timelike-singular background and the accumulated contribution of an additional effective matter sector.}
The derivatives of \(\Phi\) are
\begin{equation}
\Phi'(r)=8\pi r^2 \rho(r)-1,
\label{Phiprime}
\end{equation}
\begin{equation}
\Phi''(r)=8\pi r\left(2\rho+r\rho'\right).
\label{Phipp}
\end{equation}

The function \(\Phi\) packages the entire horizon problem into a one-dimensional global analysis. Several general properties are immediate. \added{In this formulation the analysis does not rely on a particular microscopic interpretation of the source. The density profile \(\rho(r)\) may represent a classical anisotropic distribution, an effective contribution induced by additional fields, or the Einstein-frame description of corrections coming from a more fundamental theory. What matters for horizon generation is the way in which this effective source contributes to the mass function.}

Near the origin, condition (\ref{deltaM}) leads to 
\begin{equation}
\int_0^r \rho(s)s^2 ds = o(1), \qquad r\to 0,
\end{equation}
hence
\begin{equation}
\Phi(0^+)=-2|M|<0.
\label{phi0}
\end{equation}
Thus the geometry starts in the naked regime. \added{Equivalently, the seed contribution alone does not generate a zero of \(f(r)\); any horizon must therefore be produced by the additional matter contribution at finite radius.}

At infinity, if \eqref{deltaM} holds,
\begin{equation}
\Phi(r)=2M_{\rm eff}-r+o(1), \qquad r\to\infty,
\end{equation}
so in particular
\begin{equation}
\Phi(r)\to -\infty,
\qquad r\to\infty.
\label{phiinf}
\end{equation}
Therefore \(\Phi\) is negative both near the origin and at infinity. This means that for a horizon to exist, \(\Phi\) must become positive somewhere in between. Hence horizon formation is equivalent to the existence of a positive hump in \(\Phi\). This observation shows that the problem is not merely whether the added mass exceeds the negative Schwarzschild mass parameter asymptotically, but whether the local accumulation of matter is strong enough to lift \(\Phi\) above zero at some finite radius. \added{The horizon-generation mechanism is therefore intrinsically nonlocal: it depends on the cumulative distribution of matter, not only on the total mass.}

\begin{lemma}[Existence and nature of horizons]
\label{lemma:existence_classification}
Let $\Phi(r)$ be defined in Eq.~\eqref{Phi}, and assume that $\rho(r)\geq 0$ and $\Delta M < \infty$ so that the asymptotic behavior \eqref{asymp}--\eqref{phiinf} holds.
Then:
\begin{enumerate}
\item[(i)] A horizon exists if and only if
\begin{equation}
\sup_{r>0} \Phi(r) \geq 0.
\end{equation}
\item[(ii)] If $\sup_{r>0} \Phi(r) > 0$, then $\Phi$ admits at least two distinct zeros. Generically, these zeros are simple, i.e.
\begin{equation}
\Phi(r_i)=0, \qquad \Phi'(r_i)\neq 0.
\end{equation}
\item[(iii)] If $\sup_{r>0} \Phi(r) = 0$, then there exists a unique $r_H$ such that
\begin{equation}
\Phi(r_H)=0, \qquad \Phi'(r_H)=0,
\end{equation}
and this zero is of multiplicity two. In this case, the horizon is degenerate.
\end{enumerate}
\end{lemma}
\begin{proof}
From Eqs.~\eqref{phi0} and \eqref{phiinf}, the function $\Phi(r)$ is negative both near the origin and as $r \to \infty$. Since $\Phi$ is continuous, it can admit zeros only if it reaches non-negative values somewhere. This establishes (i).
If $\sup \Phi > 0$, then $\Phi$ becomes strictly positive in some interval. Since $\Phi$ is negative near $r=0$ and as $r\to\infty$, continuity implies that $\Phi$ must cross zero at least twice. This proves (ii).
If $\sup \Phi = 0$, then the maximum is attained at some $r_H$ with $\Phi(r_H)=0$. Since $r_H$ is a point of maximum, one must have $\Phi'(r_H)=0$. Moreover, since $\Phi$ does not become positive anywhere, the zero cannot be simple, and therefore it is of multiplicity two. This proves (iii).
\end{proof}

Lemma \ref{lemma:existence_classification} makes clear that without further assumptions a static matter sector can naturally produce two horizons, just as in Reissner-Nordström-like geometries. Therefore, to obtain a unique \emph{outer} event horizon and no additional horizons in the asymptotically accessible region, more structure is needed. \added{This is precisely the point of the classification developed below: horizon generation is possible only for matter profiles whose accumulated mass produces the required positive hump in \(\Phi\), while the number of horizons is controlled by the number and nature of the extrema of this function.} These facts are shown in Fig. \ref{fig1}.

\begin{figure*}
    \centering
\includegraphics[width=0.32\textwidth]{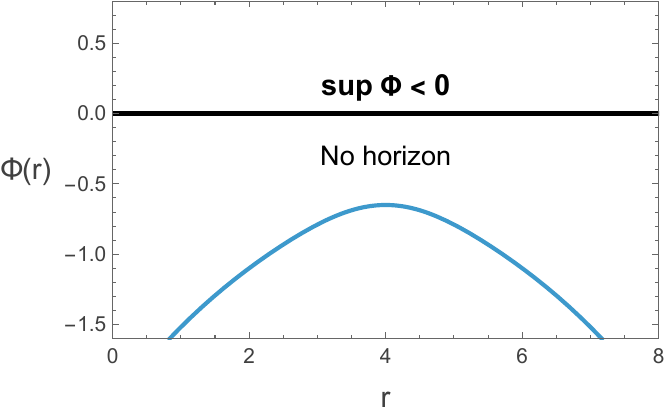}\,
\includegraphics[width=0.32\textwidth]{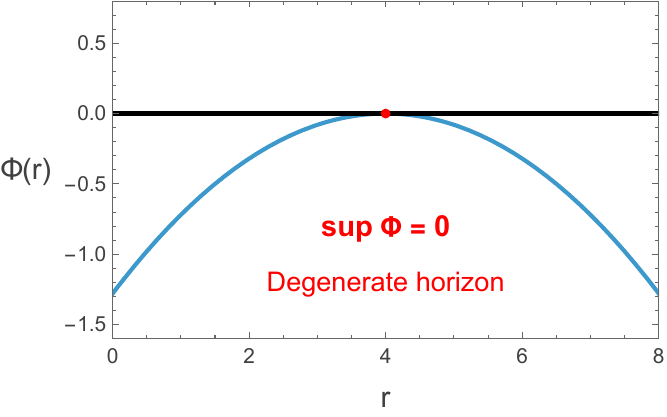}\,
\includegraphics[width=0.32\textwidth]{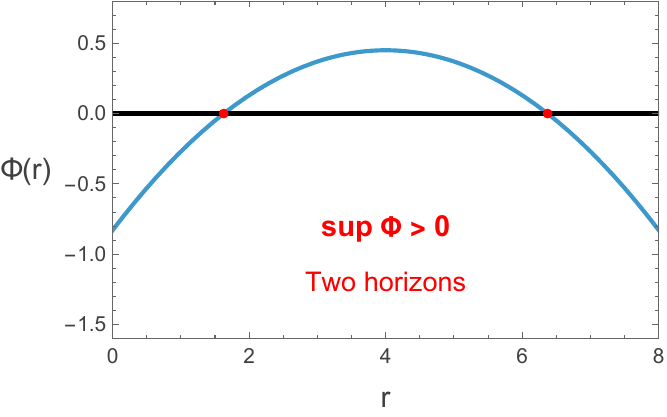}
    \caption{This panel shows all cases for the function $\Phi(r)$ given by lemma \ref{lemma:existence_classification}. }
    \label{fig1}
\end{figure*}

\begin{lemma}
\label{lem:critical}
Critical points of \(\Phi\) are solutions of
\begin{equation}
8\pi r^2\rho(r)=1.
\label{criticaleq}
\end{equation}
Hence the number of extrema of \(\Phi\) is equal to the number of intersections between the profile \(8\pi r^2\rho(r)\) and the constant value \(1\).
\end{lemma}

\begin{proof}
This follows directly from \eqref{Phiprime}.
\end{proof}

Lemma \ref{lem:critical} turns the causal problem into a profile-intersection problem. If the function \(r^2\rho(r)\) is monotonic in the relevant domain, then \(\Phi\) can have at most one extremum there, and this strongly constrains the number of horizons. \added{In this way, the existence and multiplicity of horizons are determined by simple structural properties of the effective matter profile, rather than by the detailed microscopic origin of the source.}

%========================================================
\section{Horizon-formation conditions}
\label{sec4}
%========================================================

To make precise the notion of horizon generation by an exterior matter distribution, it is natural to assume that the matter is turned on only beyond some radius \(r_c>0\). In this way the interior region \(0<r<r_c\) remains exactly the negative-mass Schwarzschild geometry, while the matter layer modifies the spacetime only outside that radius.

Thus we take
\begin{equation}
\rho(r)=0, \qquad 0<r<r_c,
\label{compactcore}
\end{equation}
and \(\rho(r)\ge 0\) for \(r\ge r_c\). Then for \(r<r_c\),
\begin{equation}
m(r)=-|M|,\qquad f(r)=1+\frac{2|M|}{r},
\end{equation}
so no horizon can appear inside the horizonless core region. The relevant question is therefore whether the exterior matter layer generates a unique \emph{outer} horizon. \added{Equivalently, we ask for sufficient conditions on the accumulated effective mass such that the metric function develops a single outer zero.}

\begin{theorem}[Sufficient horizon-formation conditions]
\label{thm:dress}
Let \(\rho(r)\) satisfy:
\begin{enumerate}
\item the total added mass is finite,
\begin{equation}
\Delta M=4\pi\int_0^\infty \rho(s)s^2 ds<\infty;
\end{equation}

\item the effective mass is positive,
\begin{equation}
M_{\rm eff}=-|M|+\Delta M>0;
\label{Meffpos}
\end{equation}

\item 
(the single-peaked compactness profile condition)
 the function
\begin{equation}
H(r)\equiv 8\pi r^2\rho(r)
\end{equation}
has a single peak and satisfies
\begin{equation}
H(r)\to0,
\qquad r\to\infty;
\end{equation}

\item there exists \(R>0\) such that
\begin{equation}
8\pi\int_0^R \rho(s)s^2ds \ge R+2|M|.
\label{compactnesscondition}
\end{equation}
\end{enumerate}

Geometrically, condition \eqref{compactnesscondition} measures whether the local accumulation of positive matter is sufficient to compensate the horizonless character induced by the negative seed contribution. In this sense, horizon formation emerges from a competition between the negative bare core and the cumulative growth of the exterior matter layer. 

Then the spacetime admits an inner and an outer horizon. If the inequality in \eqref{compactnesscondition} is strict for some \(R\), the two horizons are distinct. If the equality is saturated at the maximum of \(\Phi\), they merge into a degenerate horizon. In both cases, the outer event horizon is unique.
\end{theorem}

\begin{proof}
Since \(\rho(r)\) is regular at the origin, one has 
\begin{equation}
\int_0^r \rho(s)s^2ds=o(1),
\qquad r\to0,
\end{equation}
and therefore
\begin{equation}
\Phi(r)=-2|M|-r+o(1)<0,
\qquad r\to0.
\label{phiinnernew}
\end{equation}
Hence the geometry starts in the horizonless seed regime. \added{Thus any zero of \(f(r)\), or equivalently any zero of \(\Phi(r)\), must be generated by the exterior effective matter distribution.}

The derivative of the horizon function is
\begin{equation}
\Phi'(r)=8\pi r^2\rho(r)-1
=H(r)-1.
\end{equation}
Since \(H(r)\) has a single peak and satisfies \(H(r)\to0\) as \(r\to\infty\), it follows that \(\Phi'(r)\to-1\) as \(r\to\infty\). Therefore \(\Phi(r)\to-\infty\) at large distances.

By hypothesis, there exists \(R>0\) such that
\begin{equation}
8\pi\int_0^R \rho(s)s^2ds \ge R+2|M|.
\end{equation}
Using the definition of \(\Phi\), this is equivalent to
\begin{equation}
\Phi(R)\ge0.
\end{equation}
Since \(\Phi(r)<0\) near the origin and \(\Phi(r)\to-\infty\) as \(r\to\infty\), continuity implies that \(\Phi\) must cross zero before and after reaching its maximum. Therefore the spacetime admits an inner and an outer horizon.

Because \(H(r)\) has a single peak, the equation \(\Phi'(r)=0\) has at most two solutions. Hence \(\Phi\) possesses at most one local minimum and one local maximum. Once \(\Phi\) enters its final decreasing branch toward \(-\infty\), no additional zeros can occur. Therefore the outer horizon is unique as seen from infinity. \added{This is the sense in which the theorem provides a sufficient criterion for controlled horizon generation: the exterior observer encounters a single outer horizon and no further horizon crossings in the asymptotically accessible region.}

If the inequality is strict, the maximum of \(\Phi\) is positive and the two horizons are distinct. If the maximum just saturates zero, the two horizons merge into a degenerate horizon.
\end{proof}

\begin{remark}
The theorem is deliberately phrased in terms of a unique \emph{outer} event horizon rather than a unique positive root of \(\Phi\) on the whole half-line. This is the physically relevant notion for an exterior horizon-generation construction. What matters for the causal structure seen from infinity is the existence of a single outer horizon and the absence of additional horizons in the exterior domain. \added{The result should therefore be read as a classification statement about the zeros of \(f(r)\), not as a claim that the seed geometry itself provides a realistic physical model.}
\end{remark}

\begin{corollary}
Under the assumptions of Theorem \ref{thm:dress}, the asymptotic exterior \(r>r_H^{(+)}\) is causally simple in the sense that
\begin{equation}
f(r)>0,
\qquad
r>r_H^{(+)},
\end{equation}
where \(r_H^{(+)}\) denotes the outer horizon. Moreover, no additional outer horizons occur.
\end{corollary}

\begin{proof}
Since \(\Phi(r)\to-\infty\) as \(r\to\infty\), the function eventually enters its final decreasing branch. Because \(H(r)\) has a single peak, \(\Phi'(r)=H(r)-1\) can vanish at most twice, so \(\Phi\) possesses at most one local minimum and one local maximum. Consequently, after crossing zero at the outer horizon \(r_H^{(+)}\), the function \(\Phi\) remains negative for all larger \(r\). Therefore
\[
f(r)=1-\frac{2m(r)}{r}>0,
\qquad
r>r_H^{(+)}.
\]
Since no further zero of \(\Phi\) can occur beyond the final decreasing branch, no additional outer horizons appear.
\end{proof}

The theorem gives a clean sufficient criterion. It is also instructive to understand how horizon generation fails when its hypotheses are relaxed.

\begin{proposition}[Failure modes]
\label{prop:failure}
If any of the following occurs,
\begin{enumerate}
%\item \(r^2\rho(r)\) is not monotonic,
\item the compactness condition
\begin{equation}
8\pi\int_0^R \rho(s)s^2ds \ge R+2|M|
\end{equation}
is not satisfied for any \(R>0\);
\item \(\Delta M\le |M|\),
\end{enumerate}
then the horizon-generation mechanism may fail to produce an outer event horizon. More precisely:
\begin{itemize}
%\item non-monotonicity of \(r^2\rho(r)\) allows multiple extrema of \(\Phi\), hence multiple horizons; non-monotonicity of \(r^2\rho(r)\) allows multiple extrema of \(\Phi\); however, additional extrema do not necessarily generate additional horizons, since the corresponding oscillations of \(\Phi\) may occur entirely above or below the zero level without producing new zero crossings;
\item insufficient compactness prevents \(\Phi(r)\) from reaching non-negative values;
\item insufficient total added mass leaves the asymptotic mass non-positive and prevents a standard positive-mass Schwarzschild exterior.
\end{itemize}
\added{These failure modes emphasize that horizon generation is controlled by both global and local properties of the effective density profile: the total mass fixes the asymptotic behavior, while the compactness condition determines whether a zero of \(f(r)\) can actually be produced at finite radius.}
\end{proposition}

\begin{proof}
Each claim follows directly from the asymptotic behavior of \(\Phi\). If the compactness condition is not satisfied for any \(R>0\), then
\[
8\pi\int_0^R \rho(s)s^2ds < R+2|M|
\]
for all \(R>0\). Equivalently, $\Phi(R)<0$ for all \(R>0\), so no horizon forms. Finally, if \(\Delta M\le |M|\), then \(M_{\rm eff}\le0\), and the asymptotic mass is non-positive. In this case the asymptotic geometry is not a standard positive-mass Schwarzschild exterior.
\end{proof}

%========================================================
\section{Representative horizon-generating profiles}
\label{sec5}
%========================================================

In this section we illustrate the horizon-generation mechanism using representative matter profiles satisfying different subsets of the conditions established in Theorem~\ref{thm:dress}. The purpose of these examples is not to construct new black hole geometries, many of which are already known in different contexts, but rather to exhibit explicitly how different classes of effective matter distributions modify the horizon structure of the negative-mass Schwarzschild seed. \added{As emphasized above, the seed geometry is used only as a minimal analytically controlled representative of a horizonless spacetime with a timelike singularity. The examples below should therefore be understood as tests of horizon-generating conditions on effective matter profiles.}

%========================================================
\subsection{Monotonic asymptotically flat profiles}
%========================================================

A simple analytically tractable family is provided by exterior power-law layers,
\begin{equation}
\rho(r)=
\begin{cases}
0, & 0<r<r_c,\\[1ex]
\dfrac{\rho_0}{r^n}, & r\ge r_c,
\end{cases}
\qquad \rho_0>0.
\label{powerlaw}
\end{equation}

For \(n>3\), the total added mass is finite,
\begin{equation}
\Delta M
=
\frac{4\pi\rho_0}{n-3}\,r_c^{\,3-n},
\end{equation}
and therefore the geometry is asymptotically Schwarzschild with effective mass
\begin{equation}
M_{\rm eff}
=
-|M|
+
\frac{4\pi\rho_0}{n-3}r_c^{\,3-n}.
\end{equation}

Moreover,
\begin{equation}
r^2\rho(r)=\rho_0 r^{2-n},
\end{equation}
which is strictly decreasing for all \(n>2\). Hence the monotonicity condition required by Theorem~\ref{thm:dress} is automatically satisfied throughout the asymptotically flat range \(n>3\). 

The lifting condition becomes
\begin{equation}
8\pi \rho_0 r_c^{\,2-n}>1,
\end{equation}
while positivity of the asymptotic mass requires
\begin{equation}
\frac{4\pi\rho_0}{n-3}r_c^{\,3-n}>|M|.
\end{equation}
Therefore the power-law layer generates horizons whenever
\begin{equation}
n>3,
\qquad
8\pi \rho_0 r_c^{\,2-n}>1,
\qquad
\frac{4\pi\rho_0}{n-3}r_c^{\,3-n}>|M|.
\label{powercriterion}
\end{equation}

This family provides the simplest explicit realization of the theorem: monotonic localized matter profiles naturally generate a unique outer event horizon without additional exterior trapped regions. 

A particularly instructive case is \(n=4\), for which
\begin{equation}
f(r)
=
1-\frac{2M_{\rm eff}}{r}
+\frac{8\pi\rho_0}{r^2}.
\label{n4lapse}
\end{equation}
Although this expression is algebraically similar to the Reissner--Nordström lapse function, the source is not Maxwellian. The \(r^{-2}\) term should therefore be interpreted as an effective geometric imprint of the matter profile rather than a genuine charge contribution. \added{This example illustrates how RN-like horizon structures may arise from effective anisotropic matter without assuming an electromagnetic origin.}

%========================================================
\subsection{The critical logarithmic branch}
%========================================================

The marginal case \(n=3\) is qualitatively different and must be treated separately. In this case,
\begin{equation}
\rho(r)=
\begin{cases}
0, & 0<r<r_c,\\[1ex]
\dfrac{\rho_0}{r^3}, & r\ge r_c,
\end{cases}
\end{equation}
leading to
\begin{equation}
m(r)
=
-|M|
+
4\pi\rho_0
\ln\!\left(\frac{r}{r_c}\right).
\end{equation}
The lapse function becomes
\begin{equation}
f(r)
=
1+\frac{2|M|}{r}
-
\frac{8\pi\rho_0}{r}
\ln\!\left(\frac{r}{r_c}\right).
\end{equation}

Although \(f(r)\to1\) asymptotically, the integrated mass grows logarithmically and the geometry is not asymptotically Schwarzschild in the ADM sense. The branch \(n=3\) therefore represents a critical transition between localized profiles (\(n>3\)) and slowly decaying non-localized matter distributions. 

Interestingly, logarithmic mass growth of this type also appears in several phenomenological contexts involving slowly decaying dark-matter-inspired halo profiles \cite{Li:2012zx}. From the present viewpoint, the \(n=3\) branch provides an explicit example in which horizon formation may occur even though standard ADM asymptotic behavior is lost. 

%========================================================
\subsection{Smooth horizon-generating profiles}
%========================================================

The discontinuous exterior layers considered above provide a clean realization of the theorem, but they switch on the matter distribution abruptly at \(r=r_c\). Since Einstein's equations in the Schwarzschild gauge imply
\begin{equation}
p_r(r)=-\rho(r),
\end{equation}
a discontinuity in the density also produces a discontinuity in the radial pressure. It is therefore natural to consider smooth density profiles. \added{Smooth profiles also make clear that the horizon-generation mechanism is not an artifact of thin layers or discontinuous sources.}

A broad class of regularized profiles is given by
\begin{equation}
\rho(r)=
\rho_0
\frac{r^\eta}{(r^\mu+a^\mu)^n},
\label{generalrho}
\end{equation}
with $\rho_0>0$, $a>0$, $\eta\ge0$ and $\mu>0$. Near the origin,
\begin{equation}
\rho(r)\sim r^\eta,
\end{equation}
so the density remains regular for $\eta\ge0$, while the geometry still approaches the negative-mass Schwarzschild branch close to the center,
\begin{equation}
m(r)
=
-|M|+\mathcal O(r^{\eta+3}).
\end{equation}
\added{Thus, the central seed contribution is preserved near the origin, while the smooth effective matter profile controls the possible generation of horizons at finite radius.}

The total added mass is finite whenever
\begin{equation}
n>\frac{\eta+3}{\mu},
\end{equation}
yielding
\begin{align}
M_{\rm eff}
=
-|M|
+
\frac{4\pi\rho_0}{\mu}
a^{\eta+3-\mu n}
\frac{
\Gamma\left(\frac{\eta+3}{\mu}\right)
\Gamma\left(n-\frac{\eta+3}{\mu}\right)
}
{\Gamma(n)} .
\end{align}

The horizon structure is controlled by
\begin{equation}
\Phi'(r)=H(r)-1,
\end{equation}
where
\begin{equation}
H(r)=
8\pi\rho_0
\frac{r^{\eta+2}}{(r^\mu+a^\mu)^n}.
\end{equation}

When
\begin{equation}
\mu n>\eta+2,
\end{equation}
the function \(H(r)\) develops a single maximum,
\begin{equation}
r_{\rm max}
=
a
\left(
\frac{\eta+2}{\mu n-\eta-2}
\right)^{1/\mu},
\end{equation}
which allows the lifting condition to be evaluated explicitly. Consequently, this family provides a smooth realization of the same horizon-generation mechanism described previously for discontinuous exterior layers. \added{In particular, the parameters \(\rho_0\), \(a\), \(\eta\), \(\mu\), and \(n\) control the height, width, and localization of the compactness profile \(H(r)\), and therefore determine whether \(\Phi(r)\) develops the positive hump required for horizon formation.}

A particularly interesting example is obtained for
\begin{equation}
\eta=0,
\qquad
\mu=2,
\qquad
n=\frac52,
\end{equation}
which yields
\begin{equation}
\rho(r)=
\frac{\rho_0}{(r^2+a^2)^{5/2}}.
\end{equation}
With the identifications
\begin{equation}
a=l_0,
\qquad
\rho_0=\frac{3Ml_0^2}{4\pi},
\end{equation}
the corresponding lapse function reduces to
\begin{equation}
f(r)
=
1-\frac{2Mr^2}{(r^2+l_0^2)^{3/2}},
\end{equation}
which coincides with the T-duality-inspired geometry of~\cite{Nicolini:2019irw} once the negative-mass core contribution is removed.

Thus the T-duality-inspired solution emerges naturally as a smooth limiting realization of the broader framework developed here.

\section{Conclusions}
\label{sec6}
%========================================================
In this work we addressed the inverse problem of whether a timelike naked singularity of negative-mass Schwarzschild type can acquire horizons through the presence of a classical anisotropic matter distribution. \added{The negative-mass Schwarzschild geometry was used only as the simplest static, asymptotically flat, horizonless seed with a timelike singularity. It was not assumed to represent a realistic astrophysical object or a physical endpoint of gravitational collapse.} Within a static and spherically symmetric framework, we showed that the horizon structure is completely governed by the auxiliary function
$
\Phi(r)=2m(r)-r,
$
whose zeros determine the existence and multiplicity of horizons. On this basis, we formulated sufficient horizon-formation conditions that characterize when an exterior effective matter sector generates an outer event horizon.

Our main result is that non-negative, sufficiently localized, and monotonic matter profiles can generate horizons whenever the cumulative matter contribution compensates the negative seed contribution at finite radius. Importantly, horizon formation is not controlled by the total added mass alone, but by the local radial organization of the matter sector. A matter distribution may carry enough asymptotic mass and still fail to produce a horizon if its inner compactness is too weak. In this sense, the horizon-generation mechanism is intrinsically local-global: the profile must both lift \(\Phi(r)\) near the core and maintain sufficient cumulative growth in the exterior region. \added{These conditions are formulated directly in terms of the effective density profile and therefore do not depend on a specific microscopic origin of the source. The same criteria may be applied to classical anisotropic fluids, additional matter fields, or effective stress-energy tensors arising from semiclassical or modified-gravity corrections.}

Monotonicity plays a decisive role in organizing the causal structure. When \(r^2\rho(r)\) is monotonic in the relevant domain, the resulting geometry admits a unique outer event horizon and a causally simple exterior. By contrast, non-monotonic or oscillatory profiles generically produce multiple extrema of \(\Phi(r)\), leading to multi-horizon configurations with a richer causal structure. This shows that the radial behavior of the matter profile is encoded directly in the global horizon structure of the spacetime. \added{Thus, the analysis provides a classification of matter profiles according to the number and type of horizons they can generate.}

The representative examples analyzed in this paper illustrate the versatility of the formalism. Asymptotically Schwarzschild power-law layers provide the simplest realization of the theorem, the marginal logarithmic branch \(n=3\) marks a critical transition between localized and non-localized profiles, and RN-like lapse functions emerge naturally from specific anisotropic profiles without invoking Maxwell sources. These cases demonstrate that familiar causal structures can arise purely from the geometric organization of effective matter. 

It is important to stress that the singularity is never removed in the present construction. Unlike regular black hole models based on de Sitter cores or nonlinear electromagnetic sectors, our analysis does not regularize the central divergence. The result obtained here is therefore purely causal: suitable matter distributions can generate an outer horizon, but they do not eliminate the singularity itself. \added{Consequently, the physical content of the construction lies not in making the seed singularity regular, but in identifying the minimal effective-matter conditions under which a horizon can arise.}

The analysis is restricted to static configurations and should thus be interpreted as an existence and classification result rather than a dynamical formation mechanism. Questions related to gravitational collapse, perturbative stability, quasi-normal modes, and nonlinear evolution remain open. Even so, the present work provides a controlled setting in which the relation between effective matter organization and horizon formation can be studied within classical general relativity. \added{In particular, the results may serve as benchmarks for future models in which the same effective profiles arise from more fundamental mechanisms, such as additional fields, semiclassical backreaction, or modified gravitational dynamics.}

More broadly, the results suggest that matter profiles may be classified according to their horizon-generating efficacy: some generate a single outer horizon, some fail to generate horizons, and others produce more intricate multi-horizon geometries.

\bibliography{biblio1.bib}

@article{Penrose:1964wq,
    author = "Penrose, Roger",
    title = "{Gravitational collapse and space-time singularities}",
    doi = "10.1103/PhysRevLett.14.57",
    journal = "Phys. Rev. Lett.",
    volume = "14",
    pages = "57--59",
    year = "1965"
}

@article{Penrose:1969pc,
    author = "Penrose, R.",
    title = "{Gravitational collapse: The role of general relativity}",
    doi = "10.1023/A:1016578408204",
    journal = "Riv. Nuovo Cim.",
    volume = "1",
    pages = "252--276",
    year = "1969"
}

@book{Wald1984,
  author = {Wald, Robert M.},
  title = {General Relativity},
  publisher = {University of Chicago Press},
  year = {1984}
}

@book{HawkingEllis1973,
  author = {Hawking, S. W. and Ellis, G. F. R.},
  title = {The Large Scale Structure of Space-Time},
  publisher = {Cambridge University Press},
  year = {1973}
}

@article{Joshi2007,
  author = {Joshi, Pankaj S.},
  title = {Gravitational Collapse and Spacetime Singularities},
  journal = {Cambridge University Press},
  year = {2007}
}

@article{JoshiMalafarina2011,
  author = {Joshi, P. S. and Malafarina, D.},
  title = {Recent developments in gravitational collapse and spacetime singularities},
  journal = {Int. J. Mod. Phys. D},
  volume = {20},
  pages = {2641},
  year = {2011}
}

@article{Bardeen1968,
  author = {Bardeen, J. M.},
  title = {Non-singular general-relativistic gravitational collapse},
  journal = {Proc. GR5},
  year = {1968}
}

@article{Hayward2006,
  author = {Hayward, Sean A.},
  title = {Formation and evaporation of non-singular black holes},
  journal = {Phys. Rev. Lett.},
  volume = {96},
  pages = {031103},
  year = {2006}
}

@article{Eardley1978,
    author = "Eardley, Douglas M. and Smarr, Larry",
    title = "{Time function in numerical relativity. Marginally bound dust collapse}",
    doi = "10.1103/PhysRevD.19.2239",
    journal = "Phys. Rev. D",
    volume = "19",
    pages = "2239--2259",
    year = "1979"
}

@article{Gundlach2025,
    author = "Gundlach, Carsten and Hilditch, David and Mart{\'\i}n-Garc{\'\i}a, Jos{\'e} M.",
    title = "{Critical Phenomena in Gravitational Collapse}",
    eprint = "2507.07636",
    archivePrefix = "arXiv",
    primaryClass = "gr-qc",
    month = "7",
    year = "2025"
}

@article{Casals2016,
    author = "Casals, Marc and Fabbri, Alessandro and Mart{\'\i}nez, Cristi{\'a}n and Zanelli, Jorge",
    title = "{Quantum dress for a naked singularity}",
    eprint = "1605.06078",
    archivePrefix = "arXiv",
    primaryClass = "hep-th",
    doi = "10.1016/j.physletb.2016.06.044",
    journal = "Phys. Lett. B",
    volume = "760",
    pages = "244--248",
    year = "2016"
}

@article{Nicolini:2019irw,
    author = "Nicolini, Piero and Spallucci, Euro and Wondrak, Michael F.",
    title = "{Quantum Corrected Black Holes from String T-Duality}",
    eprint = "1902.11242",
    archivePrefix = "arXiv",
    primaryClass = "gr-qc",
    doi = "10.1016/j.physletb.2019.134888",
    journal = "Phys. Lett. B",
    volume = "797",
    pages = "134888",
    year = "2019"
}

@article{Christodoulou1984,
    author = "Christodoulou, Demetrios",
    title = "{Violation of cosmic censorship in the gravitational collapse of a dust cloud}",
    doi = "10.1007/BF01223743",
    journal = "Commun. Math. Phys.",
    volume = "93",
    pages = "171--195",
    year = "1984"
}

@article{Christodoulou1999,
    author = "Christodoulou, Demetrios",
    title = "{The Instability of Naked Singularities in the Gravitational Collapse of a Scalar Field}",
    eprint = "math/9901147",
    archivePrefix = "arXiv",
    reportNumber = "Annals migration 4-2001",
    doi = "10.2307/121023",
    journal = "Annals Math.",
    volume = "149",
    number = "1",
    pages = "183",
    year = "1999"
}

@article{Dymnikova1992,
  author = {Dymnikova, Irina},
  title = {Vacuum nonsingular black hole},
  journal = {General Relativity and Gravitation},
  volume = {24},
  number = {3},
  pages = {235--242},
  year = {1992},
  doi = {10.1007/BF00760226}
}

@article{MazurMottola2004,
  author = {Mazur, Pawel O. and Mottola, Emil},
  title = {Gravitational vacuum condensate stars},
  journal = {Proceedings of the National Academy of Sciences},
  volume = {101},
  number = {26},
  pages = {9545--9550},
  year = {2004},
  doi = {10.1073/pnas.0402717101},
  eprint = {gr-qc/0407075},
  archivePrefix = {arXiv}
}

@article{VisserWiltshire2004,
  author = {Visser, Matt and Wiltshire, David L.},
  title = {Stable gravastars -- an alternative to black holes?},
  journal = {Classical and Quantum Gravity},
  volume = {21},
  number = {4},
  pages = {1135--1152},
  year = {2004},
  doi = {10.1088/0264-9381/21/4/027},
  eprint = {gr-qc/0310107},
  archivePrefix = {arXiv}
}

@article{Li:2012zx,
    author = "Li, Ming-Hsun and Yang, Kwei-Chou",
    title = "{Galactic Dark Matter in the Phantom Field}",
    eprint = "1204.3178",
    archivePrefix = "arXiv",
    primaryClass = "astro-ph.CO",
    reportNumber = "CYCU-HEP-12-04",
    doi = "10.1103/PhysRevD.86.123015",
    journal = "Phys. Rev. D",
    volume = "86",
    pages = "123015",
    year = "2012"
}

@article{Gleiser:2006yz,
    author = "Gleiser, Reinaldo J. and Dotti, Gustavo",
    title = "{Instability of the negative mass Schwarzschild naked singularity}",
    eprint = "gr-qc/0604021",
    archivePrefix = "arXiv",
    doi = "10.1088/0264-9381/23/15/021",
    journal = "Class. Quant. Grav.",
    volume = "23",
    pages = "5063--5078",
    year = "2006"
}
\bibliographystyle{elsarticle-num}

\end{document}